\documentclass[11pt]{article}
\usepackage{amssymb,amsmath,amsthm,fullpage}

\mathchardef\mhyphen="2D

\newcommand{\Xhat}{\widehat{X}}
\newcommand{\Lhat}{\widehat{L}}

\newcommand{\dist}{\mathrm{d}}
\newcommand{\R}{\mathbb{R}}
\newcommand{\e}{\varepsilon}

\newcommand{\ball}{\mathrm{ball}}

\DeclareMathOperator*{\argmin}{argmin}
\newcommand{\wfs}{\mathrm{wfs}}

\renewcommand{\because}[1]{& \left[\text{#1}\right]}

\newtheorem{theorem}{Theorem}
\newtheorem{lemma}[theorem]{Lemma}
\newtheorem{definition}[theorem]{Definition}
\newtheorem{corollary}{Corollary}[theorem]
\usepackage[all]{xy}

\title{Adaptive Metrics for Adaptive Samples}

  \author{Nicholas J. Cavanna\\
University of Connecticut\\
  \texttt{nicholas.j.cavanna@gmail.com}
 \and 
 Donald R.~Sheehy\\
 University of Connecticut\\
 \texttt{don.r.sheehy@gmail.com}
 }
\date{}
\begin{document}
  \maketitle
 \begin{abstract}
  In this paper we consider adaptive sampling's local-feature size, used in surface reconstruction and  geometric inference, with respect to an arbitrary landmark set rather than the medial axis and relate it to a path-based adaptive metric on Euclidean space.
  We prove a near-duality between adaptive samples in the Euclidean metric space and uniform samples in this alternate metric space which results in topological interleavings between the offsets generated by this metric and those generated by an linear approximation of it. 
 After smoothing the distance function associated to the adaptive metric, we apply a result from the theory of critical points of distance functions to the interleaved spaces which yields a computable homology inference scheme assuming one has Hausdorff-close samples of the domain and the landmark set.
\end{abstract}

  \section{From Surface Reconstruction to Homology Inference} % (fold)
\label{sec:introduction}

  To reconstruct a surface from a point set, one needs a sample that is sufficiently dense with respect to not just the local curvature of the surface, but also the distance to parts of the surface that are close in the embedding but far in geodesic distance.
  Otherwise, algorithms have no way of identifying which geometrically close sample points correspond to local neighborhoods in the surface.
  Adaptive sampling with respect to the so-called \emph{local feature size}, introduced by Amenta and Bern~\cite{amenta98new}, neatly characterized said ``good'' samples and was then used in many later works on surface reconstruction with topological guarantees~\cite{bookDey2007}.
  Such \emph{adaptive samples} are in contrast to \emph{uniform samples}, where a single parameter determines the density, usually chosen as the minimum of the local feature size,  which results in amuch larger sample.
  
  Later work on geometric and homological inference related the topology of unions of balls centered at a sample $\Xhat$ near the unknown set $X$ to the topology of $X$ itself.
  A union of balls with a fixed radius can be viewed as a sublevel set of the distance function to $\Xhat$.
  If we have an adaptive sample, then we would like to scale the radii of the balls as well.
  However, if the sample is adaptive with respect to a local feature size defined as the distance to an unknown set $L$, another approximation $\Lhat$ near $L$ is necessary.
  Indeed, one interpretation of some Voronoi-based surface reconstruction algorithms is that first an approximation $\Lhat$ to the medial axis $L$ is computed from the Voronoi diagram of the sample $\Xhat$ of the unknown surface $X$.
  
  We present a new perspective on adaptive samples.
  For \emph{any} pair of disjoint, compact sets $X$ and $L$, we define a metric on $\R^d\setminus L$ with the property that a uniform sample of $X$ in the new metric corresponds to an adaptive sample in the Euclidean metric.
  This metric can be viewed as a smoothing of an adaptive metric used by Clarkson~\cite{clarkson06building} and this formulation has connections to recent work on path planning~\cite{agarwal16efficient,wein08planning} and density-based distances~\cite{cohen15approximating}.
  The main motivation is establishing a connection between adaptive sampling theory and the critical point theory of distance functions used extensively to prove topological guarantees in topological data analysis~\cite{chazal09sampling, chazal08smooth,grove93critical}.
  The latter theory gives natural topological equivalences between sublevel sets of smooth Riemannian distance functions.
By considering  these topological equivalences with respect to constructed interleavings  between the sublevel sets of a smooth adaptive metric distance function and the unions of Euclidean balls constructed from approximations to $X$ and $L$ we are able to to provide a computable homology inference method for the domain.
 
  \section{Background} % (fold)
\label{sec:background}

For the duration of this section, let $L$ and $X$ denote disjoint subsets of $\R^d$ that are compact with respect to the standard Euclidean topology.
For $x,y\in\R^d$, denote by $\text{Path}(x,y)$ the set of continuous paths from $x$ to $y$, parametrized by Euclidean arc-length.
In a similar fashion, denote by $\text{Path}(x,Y)$ the set of continuous paths from $x$ to any $y\in Y$.
For any compact set $L\subset \R^d$, define the function $f_L(\cdot): \R^d\rightarrow \R$ for $x\in\R^d$ as $f_L(x):= \min_{\ell\in L}\|x-\ell\|$.

\begin{definition} For compact $L\subset \R^d$, and $x,y\in \R^d\setminus L$, the \emph{adaptive metric} with respect to $L$ is
\[
\dist^L(x,y):= \inf_{\gamma\in\text{Path}(x,y)} \int\limits_{z\in\gamma} \frac{dz}{f_L(z)}.
\]
\end{definition}
This is only well-defined over $\R^d\setminus L$ as $f_L(l)=0$ for all $l\in L$ so any path achieving a finite integral is entirely contained in this region, and notably this implies the paths are bounded away from $L$.
We leave it to the reader to confirm that this is indeed a metric.
 The length of a unit-speed path $\gamma:[0,a]\rightarrow \R^d$ is denoted as $|\gamma|:=\int_\gamma dz = \int^{a}_0 dt$. 
 We use the descriptor adaptive for the metric as the definition naturally describes a collection of metrics dependent on the choice of $L$.
 The adaptive metric notably coincides with the minimum cost of the paths between two given points as used in the robot motion-planning works of Wein et al. and later Agarwal et al.~\cite{wein08planning,agarwal16efficient}.
 In particular, the cost of a path is the line integral of the inverse of the ``clearance" of each point along the path to the obstacle, where the obstacles in consideration are the polygons in the plane they consider for their obstacles.
Visualizing (and computing) the optimal paths becomes quite complicated as $L$ becomes more irregular.
 In~\cite{wein08planning} it is proven that when $L$ is a single point in $\R^2$ the minimal paths are logarithmic spirals, and when $L$ is a line segment the minimal paths are circular arcs.

For $y\in\R^d\setminus L$, define $f_X^L(y): = d^L(y,X)=\min_{x\in X}\dist^L(y,x)$, and $\widehat{f^L_X}(y):=\min_{x\in X}\frac{\|y-x\|}{f_L(x)}$. Note that $f_X^L(\cdot)$ is a proper distance function and $\widehat{f^L_X}(\cdot)$ is not, however the latter can be viewed as a first-order approximation of the former.
The distance function $f_X^L(\cdot)$ is not smooth in general, but a smooth Riemannian distance function can be constructed that approximates it arbitrarily well (see Section~\ref{sec:smoothing_adaptive}), however it is a $1$-Lipschitz function as shown in Lemma~\ref{lem:lipschitz}.
A function $f$ between two metric spaces $(X,\dist_X)$ and $(Y,\dist_Y)$ is said to be $K$-Lipschitz for some constant $K>0$, if for all $x,y\in X$, $\dist_Y(f(x),f(y))\leq K\dist_X(x,y)$.

\begin{lemma}\label{lem:lipschitz}
$f^L_X(\cdot)$ is a $1$-Lipschitz function from $(\R^d\setminus L, \dist^L)$ to $\R$.
\end{lemma}

\begin{proof} 
Consider any $a,b\in\R^d$ and $\e>0$.
By definition there exists $x\in X$ and $\gamma_1\in\text{Path}(a,x)$ such that $f^L_X(a)\leq \int_{\gamma_1}\frac{dz}{f_L(z)}\leq f^L_X(a)+\frac{\e}{2}$.
 Likewise, there exists $\gamma_2\in\text{Path}(a,b)$ such that $\dist^L(a,b)\leq \int_{\gamma_2}\frac{dz}{f_L(z)}\leq \dist^L(a,b)+\frac{\e}{2}$.
By construction the concatenation of the two paths $\gamma_1*\gamma_2$ is in $\text{Path}(b,X)$ with an appropriate reparameterization, so for all $\e> 0$, $f^L_X(b)\leq \int_{\gamma_1 * \gamma_2}\frac{dz}{f_L(z)}\leq f^L_X(a)+\dist^L(a,b) + \e$ so in fact $f^L_X(b) \leq f^L_X(a)+\dist^L(a,b)$.
 By symmetry we have $|f^L_X(b)-f^L_X(a)| \leq \dist^L(a,b)$.
\end{proof}

As both $f^L_X$ and $\widehat{f^L_X}$ are real-valued functions, they generate sublevel sets, also known as offsets in the topological data analysis.
The sublevel sets of former are the true offsets generated by the adaptive metric $\dist^L$ while the latter's are approximations of the true offsets.
\begin{definition}
For any compact set $L\subset \R^d$ and compact $X \subset \R^d\setminus L$ the adaptive \emph{$\alpha$-offsets} with respect to $\dist^L$ are
\[
A^L_X(\alpha) := \{x\in\R^d \mid f^L_X(x)\leq\alpha\}.
\]
\end{definition}

\begin{definition}
For any compact set $X\subset\R^d\setminus L$, for some compact set $L\subset\R^d$, the \emph{approximate $\alpha$-offsets} with respect to $\dist^L$ are
\[
B^L_X(\alpha) := (\widehat{f^L_X})^{-1}[0,\alpha] =\bigcup_{x\in X}\ball(x,\alpha f_L(x)).
\]
\end{definition}

Note the approximate offsets can also be expressed as the union of metric balls of varying radii.
One of the basic goals of this work is to relate the adaptive offsets with respect to $X$ and $L$ to the approximate offsets with respect to $\Xhat$ and $\Lhat$ via topological interleaving, where $\Lhat$ and $\Xhat$ are samples of $L$ and $X$ respectively.
Once this achieved one may apply the Persistent Nerve Lemma to the interleaving as each $B^{\Lhat}_{\Xhat}(\alpha)$ at each scale is a collection of Euclidean balls, leading to a computational homology inference method.

We can extend an adaptive metric $\dist^L$  to the associated Hausdorff distance, which is a measure of the dissimilarity between two subsets of the metric space are.
\begin{definition}
The \emph{Hausdorff distance} between two compact sets $X, Y\in (\R^d\setminus L,\dist^L)$ is defined as 
\[
\dist^L_H(X,Y):=\max\{\max_{x\in X}f^L_Y(x), \max_{y\in Y}f^L_X(y)\}
\]
\end{definition}

The Hausdorff distance between a space and a sample is a measure of the quality of the sample, namely the uniformity of it with respect to the space it is sampled from.
In this paper when we say a uniform sample, we mean one that is Hausdorff-close to the original space.
By assuming a bound on the Hausdorff distance between a compact set $X$ and a point sample $\Xhat$, we can provide containments between the offsets generated by $f_{\Xhat}^L$ and $f_X^L$ for particular scales.
The constructed symmetric relationship between the offsets over all scales is an example of a so-called filtration interleaving.

\begin{lemma}\label{lem: AX_interleaving}
Consider compact $\Xhat, X\subseteq \R^d\setminus L$ be such that $\dist_H^L(\Xhat,X)\leq \delta$.
 Then for all $\alpha\geq 0$,  $A_X^L(\alpha)\subseteq  A_{\Xhat}^L (\alpha+\delta)$ and $A_{\Xhat}^L(\alpha)\subseteq A_X^L(\alpha+\delta)$.
 
\end{lemma}

\begin{proof}
Fix $y\in A_X^L(\alpha)$.
By definition $f^L_X(y)\leq\alpha$, which implies that there exists $x\in X$ such that $\dist^L(x,y)\leq\alpha$. 
$\dist_H^L(\Xhat,X)\leq\delta$ which implies that for all $x\in X$, $f^L_{\Xhat}(x)\leq \delta$.
 Now by Lemma~\ref{lem:lipschitz}, $f^L_{\Xhat}(y)\leq f^L_{\Xhat}(x)+\dist^L(x,y)\leq\delta+\alpha$, implying $y\in A_{\Xhat}^L (\alpha+\delta)$.
 By a symmetric argument, the latter assertion holds.
\end{proof}

In contrast to the aforementioned uniform samples with respect to the adaptive metric, the following definition describes our realization of adaptive sampling with respect to the Euclidean metric.

\begin{definition}
Given compact set $L\subset \R^d$ and compact sets $\Xhat\subset X\subset \R^d\setminus L$, we say that $\Xhat$ is an \emph{$\e$-sample} of $X$, for $\e\in[0,1)$, if for all $x\in X$, there exists $p\in\Xhat$ such that $\|x-p\|\leq\e f_L(x)$.

\end{definition}

Recall that the local feature size at a point on a manifold is the distance from it to the manifold's medial axis--the closure of the collection of points with more than one closest point to the manifold.
Our definition of an $\e$-sample is a generalization of the original notion of an adaptive sample defined by Amenta and Bern~\cite{amenta98new}, which only considers $L$ to be the medial axis of $X$, and also presumes that $X$ is a manifold rather than just a compact set.

Also, note that the relation between an $\e$-sample and the approximate offsets.
If one has an $\e$-sample $\Xhat$ of $X$, then for all $x\in X$, $\ball(x,\e f_L(x))\cap \Xhat\neq \emptyset$ implying $B^L_X(\e)\cap \Xhat \neq \emptyset $

% section background (end)
  \section{Adaptive Sampling} % (fold)
\label{sec:adaptive_sampling}

  In this section we prove that a uniform sample with respect to the adaptive metric corresponds to an adaptive sample with respect to the Euclidean metric and vice versa under mild parameter assumptions.
 The following lemma is an analogous correspondence for the case of two points between a bound on the adaptive distance between them and their proximity with respect to the Euclidean metric.
 The choice of $a$ being the central point in the two statements is chosen arbitrarily.
 Theorem~\ref{thm:adaptivesamples} is a direct result of this lemma.
  
  \begin{lemma}\label{lem:distance_pts}
    Let $L\subset\R^d$ be a compact set and consider $a,b\in\R^d\setminus L$.
    The following two statements hold for all $\delta \in [0,1)$.
    \begin{enumerate}
      \item[(i)] If $\dist^L(a,b) \le \delta$, then $\|a-b\|\le \frac{\delta}{1-\delta} f_L(a)$.
      \item[(ii)] If $\|a-b\| \leq \delta {f_L(a)}$, then $\dist^L(a,b)\le \frac{\delta}{1-\delta}$.
    \end{enumerate}
  \end{lemma}
  \begin{proof}
   $(i)$ Assume $\dist^L(a,b) \le \delta$.
    Given some $\e>0$, consider $\gamma\in\text{Path}(a,b)$ such that $\dist^L(a,b)\leq \int_{\gamma} \frac{dz}{f_L(z)}\leq \dist^L(a,b) +\e$.
    Note that as $\|a-b\|$ is the length of the shortest path between $a$ and $b$ in the Euclidean metric, $\|a-b\|\leq |\gamma|$.
    We then have the following inequalities resulting from this fact and $f_L$ being $1$-Lipschitz.
    \begin{align*}
      |\gamma|=\int_{\gamma}dz 
      &= (f_L(a)+|\gamma|)\int_{\gamma}\frac{dz}{f_L(a)+|\gamma|}\\
      &\leq (f_L(a)+|\gamma|)\int_{\gamma}\frac{dz}{f_L(z)}\\
      % & = (f_L(a)+|\gamma|)\dist^L(a,b)\\
      &\leq (f_L(x)+|\gamma|)(\delta+\e).
    \end{align*}
  Considering $\e$ going to $0$ and rearranging the resulting inequality we have that $|\gamma| \leq \frac{\delta}{1-\delta}f_L(x)$.
     We then conclude that $\|a-b\|\leq |\gamma| \leq \frac{\delta}{1-\delta}f_L(x)$.
    
     $(ii)$ Assume $\|a-b\|\le \delta {f_L(a)} $.
    For all points $z$ in the straight line segment $\overline{ab}$  $\|a-z\| \leq \|a-b\|$ so the following chain of inequalities holds, 
    \[
      f_L(z)\geq f_L(a) - \|a-z\| \geq f_L(a)-\|a-b\| \geq (1-\delta) f_L(a),
    \]
    leading to the desired inequality.
    \begin{align*}
      \dist^L(a,b) 
        &=\inf_{\gamma\in\text{Path}(a,b)}\int_{\gamma}\frac{dz}{f_L(z)}\\
        &\leq \int_{\overline{ab}}\frac{dz}{f_L(z)}\\
        &\leq\frac{1}{(1-\delta)f_L(a)}\int_{\overline{ab}} dz\\
        &=\frac{\|a-b\|}{(1-\delta)f_L(a)} \\
        &\leq \frac{\delta}{1-\delta}.
    \end{align*}
  \end{proof}
  
  This lemma directly leads to our theorem relating adaptive samples in the Euclidean metric to uniform samples in the adaptive metric $\dist^L$ with respect to some compact set $L$.
  
  \begin{theorem}\label{thm:adaptivesamples}
    Let $L$ and $X$ be compact sets, let $\Xhat\subset X$ be a sample, and let $\e\in [0,1)$ be a constant.
    If $\Xhat$ is an $\e$-sample of $X$ with respect to the distance to $L$, then $\dist_H^L(X,\Xhat)\le \frac{\e}{1-\e}$.
    Also, if $\dist_H^L(X,\Xhat)\le \e<\frac{1}{2}$, then $\Xhat$ is an $\frac{\e}{1-\e}$-sample of $X$ with respect to the distance to $L$.
  \end{theorem}
  \begin{proof}
  If $\Xhat$ is an $\e$-sample of $X$, given $x\in X$, there exists $p\in\Xhat$ such that $\|x-p\|\leq \e f_L(x)$.
  By Lemma~\ref{lem:distance_pts}, $\dist^L(x,p)\leq \frac{\e}{1-\e}$, so for all $x\in X$, $f^L_{\Xhat}(x)\leq\frac{\e}{1-\e}$.
  As $\Xhat\subseteq X$, this proves $\dist^L_H(\Xhat, X)\leq\frac{\e}{1-\e}$.
  
  If $\dist^L_H(\Xhat, X)\leq\e<\frac{1}{2}$ then for all $x\in X$, $f^L_{\Xhat}(x)\leq\e$, thus there exists $p\in \Xhat$ such that $\dist^L(x,p)\leq\e$, and thus by Lemma~\ref{lem:distance_pts}, $\|x-p\|\leq\frac{\e}{1-\e}f_L(x)$. 
  As $\e<\frac{1}{2}$, $\frac{\e}{1-\e}<1$, so $\Xhat$ is an $\frac{\e}{1-\e}$-sample of $X$.

  \end{proof}

% section adaptive_sampling (end)
  \section{Interleavings} % (fold)
\label{sec:interleaving}

  A \emph{filtration} is an increasing sequence of sets or topological spaces $\mathcal{F}= (F^\alpha)_{\alpha\geq 0}$ where for all $\alpha$, $F(\alpha)\subset \R^d$ and $F(\alpha)\subseteq F(\beta)$ iff $\alpha\leq \beta$.
   We specifically consider filtrations that are generated by the sub-level sets of a real-valued function $f:\R^d\to \R$, i.e. the filtration $\mathcal{F}$ whose set at scale $\alpha$ is  defined as
  \[
    F(\alpha) := \{x\in \R^d \mid f(x)\le \alpha\} = f^{-1}[0,\alpha].
  \]
Interleavings provide a concrete relationship between two filtrations' sets.
  In topological data analysis, researchers primarily focus on or utilize interleavings that are symmetric, ones where $h_1=h_2$, and are interleaved over either the intervals $(-\infty,\infty)$, $(0,\infty)$ or $[0,\infty)$.
  Relaxing what is considered an interleaving allows us to define specific sampling conditions for which the homology inference result is valid.
The following is a generalization of the standard notion of a (symmetric) interleaving, in which we allow for asymmetry and restrictions of the intervals over which the filtrations' elements are interleaved.

  \begin{definition}\label{def:interleaved}
    A pair of filtrations \emph{$(\mathcal{F},\mathcal{G})$ is $(h_1, h_2)$-interleaved} on an interval $(s,t)$ if $F(r)\subseteq G(h_1(r))$ whenever $r,h_1(r)\in (s,t)$ and $G(r)\subseteq F(h_2(r))$ whenever $r,h_2(r)\in (s,t)$. 
    We require that the functions $h_1,h_2$ be non-decreasing over the interval $(s,t)$.
  \end{definition}
  
    Proving the existence of an interleaving between filtrations, explicitly or implicitly, is used in topological data analysis to provide insight into the topological and geometric differences (and similarities) between the filtrations when one filtration is generated by a particularly nice function.
    This idea, as it applies to this work, will be expanded upon and utilized in Section~\ref{sec:smoothing_adaptive} by considering the smoothing of $f_X^L$.

The following lemma gives us an iterative way to combine pairs of interleavings over the intersections of their interleaving intervals and will ultimately be used to construct the desired relationship.
  
  \begin{lemma}\label{lem:composinginterleavings}
    If $(\mathcal{F},\mathcal{G})$ is $(h_1,h_2)$-interleaved on $(s_1,t_1)$, and $(\mathcal{G},\mathcal{H})$ is $(h_3,h_4)$-interleaved on $(s_2,t_2)$, then $(\mathcal{F},\mathcal{H})$ is $(h_3\circ h_1,h_2\circ h_4)$-interleaved on $(s_3, t_3)$, where $s_3 = \max\{s_1,s_2\}$ and  $t_3 = \min\{t_1,t_2\}$.
  \end{lemma}
  \begin{proof}
    Given $r, h_3(h_1(r))\in (s_3,t_3)$,  we have $F(r)\subseteq G(h_1(r))\subseteq H(h_3(h_1(r)))$.
    Similarly, given $r, h_2(h_4(r))\in (s_3,t_3)$, we have $H(r)\subseteq G(h_4(r)) \subseteq F(h_2(h_4(r)))$.
  \end{proof}

    For the rest of this section, let $\Lhat\subseteq L\subset \R^d$ and $\Xhat\subseteq X \subset \R^d\setminus L$ be compact sets, with $\Lhat$ and $\Xhat$ representing samples of $L$ and $X$ respectively. 
    The desired relationship between the adaptive offset filtration $\mathcal{A}_X^L$ and the approximate offset filtration $\mathcal{B}_{\Xhat}^{\Lhat}$ will be provided by an interleaving that is built up by multiple applications of Lemma~\ref{lem:composinginterleavings} to the interleavings constructed in the remainder of this section.
  
  \subsection{Approximating $X$ with $\Xhat$} % (fold)
  \label{sub:approximating_x_with_xhat}
  
    \begin{lemma}\label{lem:XtoXhat}
      If $\dist_H^L(\Xhat,X)\le\e$, then $(\mathcal{A}_X^L, \mathcal{A}_{\Xhat}^L)$ is $(h_1, h_1)$-interleaved on $(0, \infty)$, where $h_1(r) = r + \e$.
    \end{lemma}
    \begin{proof}
      This lemma is identical to Lemma~\ref{lem: AX_interleaving} expressed in our interleaving notation.
    \end{proof}

  % subsection approximating_x_with_xhat (end)

  \subsection{Approximating the Adaptive Metric} % (fold)
  \label{sub:approximating_the_induced_metric}
    Next we show that we may reasonably approximate the sublevel sets of $f_X^L$ by Euclidean balls, which are much easier to work with than arbitrary sublevel sets, particularly when computing intersections.
    These results may be viewed as an extension of the adaptive sampling results of the previous section, Lemma~\ref{lem:distance_pts} and Theorem~\ref{thm:adaptivesamples}.
    
    \begin{lemma}\label{lem: AX_BX_interleaving}
Given compact set $L\subset \R^d$, and compact set $X\subset\R^d\setminus L$, for $r\in[0,1)$, $A_X^L(r)\subseteq B_X^L(\frac{r}{1-r})$, and for $r\in[0,\frac{1}{2})$, $B_X^L(r)\subseteq A_X^L(\frac{r}{1-r})$.
\end{lemma}

\begin{proof}
Consider $r\in [0,1)$ and $y\in A_X^L(r)$ such that $f_X^L(y)\leq r$.
By definition there exists $x\in X$ such that $\dist^L(x,y)\leq r$.
By Lemma~\ref{lem:distance_pts}, this implies that $\|x-y\|\leq\frac{r}{1-r}f_L(x)$, which implies that $y\in B_X^L(\frac{r}{1-r})$.

Now consider $r\in[0,\frac{1}{2})$ and $y\in B_X^L(r)$.
By definition, $y\in\ball(x,r f_L(x))$ for some $x\in X$ so $\|x-y\|\leq r f_L(x)$.
Applying Lemma~\ref{lem:distance_pts}, we have then have that $\dist^L(x,y)\leq \frac{r}{1-r}$, and as $f_X^L(y)\leq\dist^L(x,y)$, $y\in A_X^L(\frac{r}{1-r})$.
\end{proof}

    \begin{corollary}\label{cor:AtoB}
      The pair $(\mathcal{A}_{\Xhat}^L, \mathcal{B}_{\Xhat}^L)$ are $(h_2, h_2)$-interleaved on $(0,\frac{1}{2})$, where $h_2(r) = \frac{r}{1-r}$.
    \end{corollary}
    \begin{proof}
    This follows from considering Lemma~\ref{lem: AX_BX_interleaving} with respect to the interleaving notation.

    \end{proof}

  % subsection approximating_the_induced_metric (end)

  \subsection{Approximating $L$ with $\Lhat$} % (fold)
  \label{sub:approximating_l_with_lhat}

    A landmark set $L$ is often only approximate-able as it is frequently dependent on $X$ whose shape is object of interest in the first place.
    One may only be able to construct a finite point set sampled from it.
    For example, in the case where $L$ is the medial axis of $X$ there are several known techniques for approximating $L$, e.g. taking some vertices of the Voronoi diagram~\cite{amenta98new,bookDey2007}.
    By default this uncertainty prevents accurate evaluation of $f_L$ and by extension $\dist^L$.
    We would like to provide some sampling conditions that allow us to reasonably infer information about $L$, and said functions, by only looking at a sample $\Lhat$.
    
    Interestingly, the sampling conditions we use for $\Xhat$ are dual to those used for $\Lhat$.
    Specifically we assume an upper-bound on $\dist_H^{\Xhat}(L,\Lhat)$, or alternatively by  Theorem~\ref{thm:adaptivesamples}, $\Lhat$ must be an adaptive sample of $L$ with respect to the distance to $\Xhat$.
  
    \begin{lemma}\label{lem:LtoLhat}
      If $\dist_H^{\Xhat}(L,\Lhat)\le \delta< 1$, then $(\mathcal{B}_{\Xhat}^L, \mathcal{B}_{\Xhat}^{\Lhat})$ is $(h_3,h_3)$-interleaved on $(0,\infty)$, where $h_3(r) = \frac{r}{1-\delta}$.    
    \end{lemma}
    \begin{proof}
      Begin with arbitrary $r\in (0,\infty)$ and $x\in B_{\Xhat}^L(r)$.
      There is a point $p\in \Xhat$ such that $\frac{\|x-p\|}{f_L(p)}\le r$ and there is also a closest point $z\in \Lhat$ to $p$, because $\Lhat$ is compact, so that $f_{\Lhat}(p) = \|p-z\|$.
      Theorem~\ref{thm:adaptivesamples} and our assumption that $\dist_H^{\Xhat}(L,\Lhat)\le \delta$ together imply that there exists $y\in L$ such that 
      \begin{equation}\label{eq:LtoLhat1}
        \|y-z\| \le \frac{\delta}{1-\delta}f_{\Xhat}(z).
      \end{equation}
      We also know by the definition of the distance-to-set function that
      \begin{equation}\label{eq:LtoLhat2}
        f_{\Xhat}(z)
        = \min_{q\in \Xhat} \|z-q\|
        \le \|z-p\|
        = f_{\Lhat}(p),
      \end{equation}
      so we can relate $f_L(p)$ to $f_{\Lhat}(p)$ as follows
      \begin{align*}
        f_L(p)
          &\le \|y-p\|\because{$y\in L$}\\
          &\le \|y-z\| + \|z-p\|\because{triangle inequality}\\
          % &\le \frac{\e}{1-\e}f_{\Xhat}(z) + \|z-p\|\because{by~\eqref{eq:LtoLhat1}}\\
          &\le \frac{1}{1-\delta}f_{\Lhat}(p).\because{by~\eqref{eq:LtoLhat1} and \eqref{eq:LtoLhat2}}.
      \end{align*}
      Collectively the following holds,
      \begin{align*}
        \frac{\|x-p\|}{f_{\Lhat}(p)}
          &\le \frac{\|x-p\|}{(1-\delta)f_L(p)}
          \le \frac{r}{1-\delta} = h_3(r),
      \end{align*}
      therefore $x\in B_{\Xhat}^{\Lhat}(h_3(r))$ so we conclude that $B_{\Xhat}^L(r) \subseteq B_{\Xhat}^{\Lhat}(h_3(r))$.
      The proof is symmetric o show that $B_{\Xhat}^{\Lhat}(r) \subseteq B_{\Xhat}^L(h_3(r))$
    \end{proof}

  % subsection approximating_l_with_lhat (end)

  \subsection{Putting it all together} % (fold)
  \label{sub:putting_it_all_together}
    We can now combine all the previous interleaving results using Lemma~\ref{lem:composinginterleavings} to arrive at our penultimate theorem which establishes an interleaving between the approximate offsets filtration for the approximate spaces to the adaptive metric offsets filtration for the true spaces.

    \begin{theorem}\label{thm:biginterleaving}
      Let $\Lhat\subseteq L\subset \R^d$ and $\Xhat\subseteq X \subset \R^d\setminus L$ be compact sets.
      If $\dist_H^{\Xhat}(L,\Lhat)\le \delta< 1$ and $\dist_H^L(\Xhat,X)\leq\e<1$, then $(\mathcal{A}_X^L, \mathcal{B}_{\Xhat}^{\Lhat})$ are $(h_4, h_5)$-interleaved on $(0,1)$, where $h_4(r) = \frac{r+\e}{(1-r-\e)(1-\delta)}$ and $h_5(r) = \frac{r}{1-\delta-r}+\e$.
    \end{theorem}

\begin{proof}
Applying Lemma~\ref{lem:composinginterleavings} to the interleavings from Lemma~\ref{lem:XtoXhat} and Corollary~\ref{cor:AtoB}, we have that $(\mathcal{A}_X^L, \mathcal{B}_{\Xhat}^L)$ is $(h_2\circ h_1, h_1\circ h_2)$-interleaved on $(0,1)$.
This interleaving combined with that from Lemma~\ref{lem:LtoLhat} yields that  $(A_X^L,B_{\Xhat}^{\Lhat})$ is $(h_3\circ h_2\circ h_1, h_1\circ h_2\circ h_3)$ interleaved on $(0, 1)$.
Now we simply must compute $h_3\circ h_2\circ h_1$ and $h_1\circ h_2\circ h_3$ as follows.

\begin{align*}
(h_3\circ h_2\circ h_1)(r)=(h_3\circ h_2)(r+\delta)&=h_3(\frac{r+\delta}{1-r-\delta})\\
&=\frac{r+\delta}{(1-r-\delta)(1-\e)}\\
(h_1\circ h_2\circ h_3)(r)=(h_1\circ h_2)(\frac{r}{1-\e})&=h_1(\frac{r}{(1-\e)(1-\frac{r}{1-\e})})\\
&=h_1(\frac{r}{1-\e-r})\\
&=\frac{r}{1-\e-r}+\delta\\
\end{align*}

This computation results in our interleaving functions being $h_4(r) =\frac{r+\delta}{(1-r-\delta)(1-\e)} $ and \\$h_5(r) = \frac{r}{1-\e-r}+\delta$.
\end{proof}
  % subsection putting_it_all_together (end)

% section interleaving (end) 
  \section{Critical Points of Distance Functions} % (fold)
\label{sec:critical_points}

  Here we give a minimal presentation of the critical point theory of distance functions that motivates the need for interleavings of sublevel sets of the distance functions we consider.
 
  Given a smooth Riemannian manifold $M$ and a compact subset $X\subset M$ consider the function $f_X:M\to \R$ that maps each point in $M$ to the distance to its nearest point in $X$ as determined by the metric on the manifold.
  The gradient of $f_X$ is well-defined on $M$ and its critical points are those points on which the gradient evaluates to $0$.
  The critical values of $f_X$ are the values $r\in \R_{\geq 0}$ such that $f_X^{-1}\{r\}$ contains a critical point.
  The critical point theory of distance functions developed by Grove~\cite{grove93critical} and others extends ideas from Morse theory to these distance functions.  
  In particular, the theory provides us the following result.

  \begin{lemma}\label{lem:isotopy}
    If $[r,r']$ contains no critical values then $f_X^{-1}[0,r] \hookrightarrow f_X^{-1}[0,r']$ is a homotopy equivalence.
  \end{lemma}
  
  This implies that for intervals that don't contain critical values, the inclusion maps between elements of the filtration on those intervals are all all homotopy equivalences and therefore induce natural isomorphisms at the homology level.
  We will use this in the next section to infer information about the homology of filtrations that are interleaved with such a filtration generated by a Riemannian distance function.

% section critical_points (end)
  
\section{Smooth Adaptive Distance and Homology Inference} % (fold)
\label{sec:smoothing_adaptive}

We will now introduce the smoothed distance function $\widetilde{f_X^L}$ and use it in conjunction with the results proved in Section~\ref{sec:critical_points} to provide a method to infer the homology of the so-called smooth adaptive offsets by looking solely at the approximate offsets.

For a compact set $L\subset \R^d$ and $\beta\geq 0$  denote by $L^\beta:=\{x\in \R^d\mid \min_{y\in L} \|x-y\|\leq \beta\}$ the offsets of $L$ with respect to the Euclidean metric.
The following lemmas gives upper and lower bounds on the value of a smoothing of the distance-to-set function $f_L$, $\widetilde{f_L}$, which is defined on an arbitrarily smaller subset of Euclidean space.
\begin{lemma}\label{lem:smooth}
Consider a compact set $L\subset \R^d$. Given $\alpha\in (0,1)$, for all $\beta\in (0,1)$, there exists smooth function $\widetilde{f_L}:\R^d\setminus L^\beta \rightarrow \R$ such that for all $x\in \R^d\setminus L^\beta$, $(1-\alpha)f_L(x) < \widetilde{f_L}(x) <(1+\alpha)f_L(x)$.

\end{lemma}

\begin{proof}
By a result from ~\cite{greene79approximations}, for all $\e>0$, there exists a smoothing $\widetilde{f_L}:\R^d\setminus L^\beta \rightarrow \R$ of the distance function $f_L$ such that $\|f_L-\widetilde{f_L}\|_{\infty}<\e$.
Choose $\e=\beta\alpha$, for the given $\alpha\in (0,1)$.
By the approximation property of $\widetilde{f_L}$, for all $x\in\R^d\setminus L^\beta$ we have that $f_L(x)-\e <\widetilde{f_L}(x) <f_L(x)+\e$.
Also note that for all $x\in \R^d\setminus L^\beta$, $f_L(x)>\beta=\frac{\e}{\alpha}$ and thus $\alpha f_L(x)>\e$.
Combining the aforementioned we have that $f_L(x)(1-\alpha)< f_L(x)-\e$ and $f_L(x)+\e< f_L(x)(1+\alpha)$.
\end{proof}

Consider $\widetilde{f_L}$ as defined in Lemma~\ref{lem:smooth}.
Using this we can define a smooth adaptive distance function $\widetilde{f^L_X}$ and provide upper and lower bounds on its value with respect to the original adaptive distance function $f_X^L$.
For $x,y\in\R^d\setminus L^\beta$, we define 
\[
\widetilde{\dist^{L}}(x,y):= \inf_{\gamma\in\text{Path}(x,y)} \int_\gamma \frac{dz}{\widetilde{f_L}(z)}
\]
and $\widetilde{f_X^L}(y):= \widetilde{d^L}(y, X)$.

\begin{lemma}\label{lem:smooth_distance}
Given  $\alpha,\beta \in (0,1)$ and a smooth function $\widetilde{f_L}$ defined on $\R^d\setminus L^\beta$ as constructed in the proof of Lemma~\ref{lem:smooth}, consider  a compact set $X\subset \R^d\setminus L^\beta$. The Riemannian distance function $\widetilde{f^L_X}(\cdot):=\widetilde{\dist^L}(\cdot,X)$ satisfies the following property for all $y\in \R^d\setminus L^\beta$, 
\[
\frac{1}{1+\alpha}f_X^L(y)< \widetilde{f_X^L}(y)< \frac{1}{1-\alpha}f_X^L(y).
\]
\end{lemma}

\begin{proof}
%TODO: $\widetilde{f^L_X}(\cdot)$ is a Riemannian distance function.
Given two points $x,y\in \R^d\setminus L^\beta$, and any $\e > 0$, consider $\gamma,\gamma' \in\text{Path}(x,y)$ such that $\dist^L(x,y)\leq \int_\gamma \frac{dz}{f_L(z)} \leq \dist^L(x,y) +\e$ and  $\widetilde{\dist^L}(x,y)\leq \int_{\gamma'} \frac{dz}{\widetilde{f_L}(z)} \leq \widetilde{\dist^L}(x,y) + \e$.
We then have the following inequalities resulting from inverting the inequalities in Lemma~\ref{lem:smooth}.
\[
\widetilde{\dist^L}(x,y)\leq \int_{\gamma}\frac{dz}{\widetilde{f_L}(z)} < \frac{1}{1-\alpha} \int_{\gamma}\frac{dz}{f_L(z)} \leq \frac{1}{1-\alpha}\dist^L(x,y)+\frac{\e}{1-\alpha},
\]

and

\[
\frac{1}{1+\alpha}\dist^L(x,y) \leq \frac{1}{1+\alpha}\int_{\gamma'}\frac{dz}{f_L(z)} < \int_{\gamma'}\frac{dz}{\widetilde{f_L}(z)} \leq \widetilde{\dist^L}(x,y) + \e.
\].

Since these equalities hold for all $\e>0$, then we can conclude that for all pairs $x,y\in\R^d\setminus L^\beta$, $\frac{1}{1+\alpha}\dist^L(x,y) < \widetilde{\dist^L}(x,y) < \frac{1}{1-\alpha}\dist^L(x,y)$.

Now consider $y\in \R^d\setminus L^\beta$. 
Denote $x'=\argmin_{x\in X}\dist^L(y,x)$ and $x''=\argmin_{x\in X}\widetilde{\dist^L}(y,x)$.
We remind the reader that these points' existences are guaranteed by the Extreme Value Theorem.
By examining these variables with respect to the previous inequality we know that
\[
\frac{1}{1+\alpha}\dist^L(y,x') \leq \frac{1}{1+\alpha}\dist^L(y,x'') < \widetilde{\dist^L}(y,x'')\leq\widetilde{\dist^L}(y,x') < \frac{1}{1-\alpha}d^L(y,x').
\]

By applying the definitions of both adaptive distance functions to the previous expression we obtain the desired inequality,
\[
\frac{1}{1+\alpha}f^L_X(y)<\widetilde{f^L_X}(y)<\frac{1}{1-\alpha}f^L_X(y).
\]
\end{proof}

Define the Riemannian adaptive offsets of $X$ as $\widetilde{A}_X^L(\alpha):= \{x\in \R^d\mid \widetilde{f_X^L}(x)\leq\alpha\}$, and denote the corresponding filtration by $\widetilde{\mathcal{A}}_X^L$.
The following result reestablishes Lemma~\ref{lem:smooth_distance} in the language of filtrations and establishes an interleaving of the Riemannian adaptive offsets with the original adaptive offsets.
%using Section~\ref{sec:critical_points}.

\begin{corollary}\label{cor:smooth_interleaving}
Consider a compact set $L\subset \R^d$.
 Given $\alpha,\beta\in (0,1)$, for compact $X\subset \R^d\setminus L^\beta$, there exists a Riemannian distance function $\widetilde{f_X^L}:\R^d\rightarrow \R$, such that $(\widetilde{\mathcal{A}}_X^L,\mathcal{A}_X^L)$ are $(h_6,h_7)$-interleaved on $(0,\infty)$, where $h_6(r)=(1+\alpha)r$ and $h_7(r)=\frac{r}{1-\alpha}$.
\end{corollary}

\begin{proof}
By  Lemma~\ref{lem:smooth_distance}, there exists a Riemannian distance function $\widetilde{f^L_X}:\R^d\rightarrow \R$, such that for all $y\in \R^d\setminus L^\beta$, 
\[
\frac{1}{1+\alpha}f^L_X(y)<\widetilde{f^L_X}(y)<\frac{1}{1-\alpha}f^L_X(y),
\]

so for $r\in(0,\infty)$ and $y\in \widetilde{A}^L_X(r)$, $\widetilde{f_X^L}(y)\leq r$, and thus $f_X^L(y)\leq (1+\alpha)r$, which implies that $y\in A_X^L((1+\alpha)r)$, so  $\widetilde{A}^L_X(r)\subseteq A_X^L((1+\alpha)r)$.

On the other hand, for $r\in (0,\infty)$ and $y\in A^L_X(r)$, $f_X^L(y)\leq r$, and thus $\widetilde{f^L_X}(r)\leq\frac{r}{1-\alpha}$, so $ A_X^L(r)\subseteq \widetilde{A}^L_X(\frac{r}{1-\alpha}).$
\end{proof}

Combining the previous corollary with Theorem~\ref{thm:biginterleaving} in Subsection~\ref{sub:putting_it_all_together}, we obtain an interleaving between the Riemannian adaptive offsets and the approximate offsets.
This will then allow us to apply Lemma~\ref{lem:isotopy} and standard topological data analysis techniques to this interleaving to give a method of homology inference for arbitrary small offsets of $X$ as we have a Reimannian distance function generating the smooth adaptive offsets filtration.
\begin{lemma}~\label{lem:biggerinterleaving}
Given $\alpha,\beta\in (0,1)$, consider compact sets $\Lhat\subseteq L \subset \R^d$ and compact sets $\Xhat\subseteq X \subset \R^d\setminus L^\beta$, such that $\dist_H^{\hat{X}} (L,\hat{L})\leq \delta <1$ and $\dist_H^L(\hat{X},X)\leq \e <1$, then $(\tilde{\mathcal{A}}_X^L, \mathcal{B}_{\hat{X}}^{\hat{L}})$ are $(h_8, h_9)$-interleaved on $(0,1)$, where $h_8(r)= \frac{r+\alpha r+ \e}{(1-r-r\alpha-\e)(1-\delta)}$ and $h_9(r)=\frac{r}{(1-\alpha)(1-\delta-r)} + \frac{\e}{1-\alpha}$.
\end{lemma}

\begin{proof}
 The hypotheses of the statement satisfy the hypotheses of both Theorem~\ref{thm:biginterleaving} and Corollary~\ref{cor:smooth_interleaving} so one knows that $(\mathcal{A}_X^L, \mathcal{B}_{\Xhat}^{\Lhat})$ are $(h_4, h_5)$-interleaved on $(0,1)$, where $h_4(r) = \frac{r+\e}{(1-r-\e)(1-\delta)}$, and $h_5(r) = \frac{r}{1-\delta-r}+\e$ and $(\widetilde{\mathcal{A}}_X^L, \mathcal{A}_X^L)$ are $(h_6,h_7)$-interleaved on $(0,\infty)$, where $h_6(r)=(1+\alpha)r$ and $h_7(r)=\frac{r}{1-\alpha}$.
 By applying Lemma~\ref{lem:composinginterleavings} and composing the necessary functions, we achieve the stated interleavings.
 \end{proof}
 
For clarity, we rewrite Lemma~\ref{lem:isotopy} with our definitions.
 
  \begin{lemma}\label{lem:critical_points_of_distance_functions}
    Let $X\subset \R^d$ be a compact set and let $\widetilde{f_X^L}(\cdot)=\widetilde\dist_L(\cdot, X)$ be the smooth Riemannian adaptive distance function.
    If an interval $[s,t]$ contains no critical values of $\widetilde{f_X^L}(\cdot)$, then the inclusion $\widetilde{A}_X^L(s)\hookrightarrow \widetilde{A}_X^L(t)$ is a homotopy equivalence.
  \end{lemma}
  
  Weak feature size (wfs), introduced by Chazal and Leutier in~\cite{chazal05weak}, is the least positive critical value of a Riemannian distance function.
  We denote by $\textrm{wfs}_L(X)$ the weak feature size with respect to $\widetilde{f_X^L}(\cdot)$.
  
  We now state our final theorem and its corollary which together provide a method to infer the homology in all dimensions of an arbitrarily small adaptive offset of our space $X$ of interest by looking at the homology of the inclusion between the approximate offsets of the sample $\Xhat$ with respect to $\Lhat$.
  
  \begin{theorem}\label{thm:homology_inference}
Given $\alpha,\beta\in (0,1)$, consider compact sets $\hat{L}\subseteq L \subset \R^d$ and compact sets $\hat{X}\subseteq X \subset \R^d\setminus L^\beta$, such that $\dist_H^{\hat{X}} (L,\hat{L})\leq \delta <1$ and $\dist_H^L(\hat{X},X)\leq \e <1$.
Given any $\eta>0$, such that $h_9h_8h_9h_8(\eta)<1$, if $\textrm{wfs}_L(X)>h_9h_8h_9 h_8(\eta)$, then 
\[
H_*(\widetilde{A}_X^L(\eta))\cong \textrm{im}\,(H_*(B_{\Xhat}^{\Lhat}(h_8 (\eta))) \hookrightarrow B_{\Xhat}^{\Lhat} (h_8h_9h_8(\eta))).\]
\end{theorem}

\begin{proof}

Given $\eta>0$ such that $h_9h_8h_9h_8(\eta)<1$, we have the following sequence of inclusions as a result of Lemma~\ref{lem:biggerinterleaving}.
\begin{equation}\label{eq:inclusions}
\xymatrix{
	\widetilde{A}_X^L(\eta)~\ar@{^{(}->}[r]^(.4){a} & B_{\Xhat}^{\Lhat}(h_8(\eta))~\ar@{^{(}->}[r]^(.45){b} & \widetilde{A}_X^L(h_9h_8(\eta))~\ar@{^{(}->}[r]^(.6){c} & \ldots \\
	\ldots \ar[r]^(.25){c} &B_{\Xhat}^{\Lhat}(h_8h_9h_8(\eta))~\ar@{^{(}->}[r]^(.47){d} & \widetilde{A}_X^L(h_9h_8h_9h_8(\eta))  .
	}
\end{equation}

As we assume that $\wfs_L(X)>h_9h_8h_9 h_8(\eta)$, by the definition of weak feature size, Lemma~\ref{lem:critical_points_of_distance_functions} implies that the inclusions $b\circ a$ and $d\circ c$ are homotopy equivalences. 
We remind the reader that if two spaces are homotopy equivalent, all the induced homology maps between the spaces are isomorphisms.
By applying homology to each space and inclusion in the previous sequence, we have the following sequence of homology groups, where $b_*\circ a_*$ and $d_*\circ c_*$ are isomorphisms.

\begin{equation}\label{eq:homology_inclusions}
\xymatrix{
	H_*(\widetilde{A}_X^L(\eta))~\ar[r]^(.4){a_*} & H_*(B_{\Xhat}^{\Lhat}(h_8(\eta)))~\ar[r]^(.45){b_*} & H_*(\widetilde{A}_X^L(h_9h_8(\eta)))~\ar[r]^(.7){c_*} & \\
	 & H_*(B_{\Xhat}^{\Lhat}(h_8h_9h_8(\eta)))~\ar[r]^(.47){d_*} & H_*(\widetilde{A}_X^L(h_9h_8h_9h_8(\eta))).
	}
\end{equation}

The aforementioned isomorphisms $b_*\circ a_*$ and $d_*\circ c_*$ factor through $B_{\Xhat}^{\Lhat}(h_8(\eta))$ and $B_{\Xhat}^{\Lhat}(h_8h_9h_8(\eta))$ respectively, proving that $b_*$ is surjective and $c_*$ is injective.
We then have that $H_*(\widetilde{A}_X^L(\eta))\cong H_*(\widetilde{A}_X^L(h_9h_8(\eta))) \cong \textrm{im}\,b_*\cong \textrm{im}\,(c_*\circ b_*)$.
\end{proof}

Furthermore, if we assume that our approximation $\Xhat$ of $X$ is a finite sample, then each $B_{\Xhat}^{\Lhat}(\alpha)$ at any scale is the union of a finite number of Euclidean balls.
We may then consider the good covers $\{\ball(x, \alpha f_{\Lhat}(x))\}_{x\in \Xhat }$ of  each $B_{\Xhat}^{\Lhat}(\alpha)$, and thus the corresponding good cover filtration of $\mathcal{B}_{\Xhat}^{\Lhat}$.
Note these are good covers as each element is a Euclidean balls and thus their intersections are contractible so
by the Nerve Theorem~\cite{hatcher01}, 
\[
H_*(\text{Nrv}(\{\ball(x, \alpha f_{\Lhat}(x))\}_{x\in \Xhat }))\cong H_*(B_{\Xhat}^{\Lhat}(\alpha)).
\]
As the interleaving maps are inclusions, the Persistent Nerve Lemma~\cite{chazal08towards} applies to the result of Diagram~\ref{eq:homology_inclusions} combined with the above isomorphisms.
This allows us to examine the nerves to compute the homology of the arbitrarily small adaptive offset of $X$.
Explicitly it yields the following computable inference equality.
\[
H_*(\widetilde{A}_X^L(\eta))\cong \textrm{im}\;H_*(\text{Nrv}(\{\ball(x, h_8(\eta) f_{\Lhat}(x))\}_{x\in \Xhat}))\hookrightarrow \text{Nrv}(\{\ball(x, h_8h_9h_8(\eta) f_{\Lhat}(x))\}_{x\in \Xhat })).
\]

  %-----------------------------------------
  \bibliographystyle{unsrt} 
  \bibliography{bibliography}  
 
\end{document}